\documentclass[tran,onecolumn,noversion,nonote]{cdcarticle}
\usepackage{cite}
\usepackage{graphicx}
\usepackage[cmex10]{amsmath}
\usepackage{amsthm,amssymb}
\usepackage{array}
\usepackage{enumerate}
\usepackage[hidelinks]{hyperref}
\usepackage{combelow}
\usepackage{tikz}

\makeatletter
\def\@IEEElegacywarn#1#2{}
\makeatother

\newtheorem{thm}{Theorem}
\newtheorem{lem}[thm]{Lemma}
\newtheorem{prop}[thm]{Proposition}

\newtheorem{defn}[thm]{Definition}
\newtheorem{cor}[thm]{Corollary}

\renewenvironment{proof}{\noindent{\bf Proof.}}{ \hfill~\qed}
\newenvironment{pproof}[1]{\vspace{2mm}\noindent{\bf Proof of #1.}}{ \hfill~\qed}
\def\qed{\rule[0pt]{5pt}{5pt}\par\medskip}

\newcommand{\bmat}[1]{\begin{bmatrix}#1\end{bmatrix}}

\DeclareMathOperator{\adj}{adj}
\DeclareMathOperator{\ch}{char}
\DeclareMathOperator{\delay}{delay}

\newcommand{\ideal}{\lhd}
\newcommand{\isomorphic}{\cong}

\newcommand{\p}{\textup{p}}
\renewcommand{\sp}{\textup{sp}}

\newcommand{\R}{\mathbb{R}} 
\newcommand{\Z}{\mathbb{Z}} 
\newcommand{\C}{\mathbb{C}} 
\newcommand{\Q}{\mathbb{Q}} 
\newcommand{\F}{\mathbb{F}} 

\newcommand{\set}[2]{\ensuremath{\left.\left\{ #1 \,\right\vert\, #2 \right\}}}

\newcommand{\eemph}[1]{\textbf{\textit{#1}}}

\newcommand{\mr}{\R(\mathbf{s})}

\newcommand{\mx}{ \bigl(\R(\mathbf{x})\bigr)(\mathbf{s}) }
\newcommand{\my}{ \R(\mathbf{x},\mathbf{s}) }

\def\note#1{}

\begin{document}
\title{An Algebraic Approach to the Control of Decentralized Systems}

\author{Laurent~Lessard~\and~Sanjay~Lall}
\note{}
\maketitle

\begin{abstract}
Optimal decentralized controller design is notoriously difficult, but
recent research has identified large subclasses of such problems that
may be convexified and thus are amenable to solution via efficient
numerical methods. One recently discovered sufficient condition for
convexity is \emph{quadratic invariance} (QI). Despite the simple
algebraic characterization of QI, which relates the plant and
controller maps, proving convexity of the set of achievable
closed-loop maps requires tools from functional analysis. In this
work, we present a new formulation of quadratic invariance that is
purely algebraic. While our results are similar in flavor to those
from traditional QI theory, they do not follow from that body of
work. Furthermore, they are applicable to new types of systems that
are difficult to treat using functional analysis.  Examples discussed
include rational transfer matrices, systems with delays, and
multidimensional systems.
\end{abstract}

\section{Introduction}

The problem
of designing control systems where multiple controllers are
interconnected over a network to control a collection of
interconnected plants is long-standing and
difficult~\cite{blondel,witsenhausen}. \emph{Quadratic invariance} is
a mathematical condition which, when it holds, allows one to bring to
bear the tools of Youla parameterization to find optimal
controllers~\cite{rotkowitz02,rotkowitz06}. A network system has the
requisite quadratic invariance under a surprisingly wide set of
circumstances. These include cases where the controllers can
communicate more quickly than the plant dynamics propagate through the
network~\cite{rotkowitz10}.

There is a large and diverse body of literature addressing
decentralized control theory and specifically conditions that make a
problem more tractable in some sense. The seminal work of Ho and
Chu~\cite{hochu} develops the \emph{partial nestedness} condition
under which there exists an optimal decentralized controller that is
linear. More recently, Qi, Salapaka, et al. identified many different
decentralized control architectures that may be cast as tractable
optimization problems~\cite{voulgaris_stabilization}. LMI formulations
of distributed control problems are developed
in~\cite{dulleruddandrea,langbort2004}. Stabilization was fully
characterized for all QI problems in~\cite{sabau_2011}. Explicit
state-space solutions were also found for classes of delayed
problems~\cite{lamperski_lessard}, poset-causal
problems~\cite{shah13}, and two-player output-feedback
problems~\cite{lessard_tpof}.

There have been relatively few works treating decentralized control
from a purely algebraic perspective. One recent example is the work of
Shin and Lall~\cite{shin2012decentralized}, where elimination theory
is used to express solutions to decentralized control problems as
projections of semialgebraic sets. Quadratic invariance was first
treated using an analytic framework
in~\cite{rotkowitz02,rotkowitz06}. The aim of this paper is to address
an algebraic treatment of quadratic invariance. The
consequence of this is not only a new proof of existing results in
some cases but also an extension of these results to a significantly
different class of models.  Instead of requiring analytic properties
of our system model, we will require algebraic ones. For
example,~\cite{rotkowitz06} requires that the set of allowable
controllers be a closed inert subspace, whereas in this work we
require that it be a module. The class of systems covered in this
paper includes multidimensional systems, which are not covered in
existing works. This is discussed in Section~\ref{subsec:multid}.

Many topics in control have historically been treated from both
analytic and algebraic viewpoints. As early as 1965, Kalman proposed
the use of modules as the natural framework in which to represent
linear state-space systems~\cite{kalman_1965}.  When systems are
viewed as maps on signal spaces, one has many choices. If one
represents systems as transfer functions, then one can either consider
the generality of transfer functions in a Hardy space and use analytic
methods to prove results, or one can consider formal power series or
rational functions and use algebraic methods. Often, the two
frameworks use very different proof techniques, which provide
different insights and ranges of applicability. This is a fundamental
choice in how we represent the basic objects~\cite{klarner_1969}.
This dichotomy exists in many facets of the control systems
literature. For example, spectral factorization is easily considered
from an algebraic perspective. The Riesz-Fej\'{e}r theorem states that
a trigonometric polynomial which is nonnegative on the circle may be
factored into the product of two polynomials, one of which is
holomorphic inside the disc and the other
outside~\cite{riesz_1990}. This is the fundamental algebraic version
of the discrete-time SISO spectral factorization result. For
comparison, the analytic version of this result is commonly known as
Wiener-Hopf factorization~\cite{wiener}.  Of course, the same choice
of frameworks exists beyond factorization. The theory of stabilization
as introduced by Youla~\cite{youla} was developed both
algebraically~\cite{vidyasagar} and
analytically~\cite{smith_1989}. The idea of algebraic representations
have also proven useful in areas such as realization
theory~\cite{beck_2001}, model reduction~\cite{beckdoyleglover}, and
nonlinear systems theory~\cite{fliess_1983}.

The work in this paper is based on preliminary results that first
appeared in~\cite{lessard2010algebraic,lessard_thesis}. Unlike these
early works, all invariance results in the present work include both
necessary and sufficient conditions, and all the proofs are purely
algebraic.
The paper is organized as follows. The remainder of the introduction gives an overview of quadratic invariance and existing analytic results. Invariance results are proven and discussed for matrices, rings and fields, and rational functions in Sections~\ref{sec:matrix_case},~\ref{sec:rings_and_fields}, and ~\ref{sec:rationals}, respectively. We present some illustrative examples in Section~\ref{sec:examples} and we summarize our contributions in Section~\ref{sec:conclusion}.

\subsection{Quadratic invariance}

We have adopted the notation convention from~\cite{rotkowitz02,rotkowitz06} to make the works readily comparable. Given a plant $G$, which is a map from an input space $\mathcal{U}$ to an output space $\mathcal{Y}$, we seek to design a controller $K:\mathcal{Y}\to\mathcal{U}$ that achieves desirable performance when connected in feedback with $G$. The main object of interest is the function $h:M\to M$ given by
\[
h(K) = -K(I-GK)^{-1}
\]
Here, the domain $M$ is the set of maps $K:\mathcal{Y}\to\mathcal{U}$ such that $I-GK$ is invertible. The image of $h$ is again $M$ because $h$ is an \emph{involution}. That is, $h(h(K)) = K$. We will be more specific shortly about the nature of the spaces $\mathcal{U}$ and $\mathcal{Y}$ (and consequently the maps $G$ and $K$). 

The motivation for studying $h$ is that it is a linear fractional transform that occurs in feedback control. Consider for example the four-block plant of Figure~\ref{fig:4blockplant}.

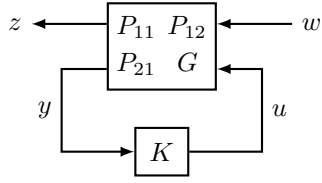
\begin{figure}[ht]
\centering
\begin{tikzpicture}[thick,auto,>=latex,node distance=1.4cm]
\tikzstyle{block}=[draw,rectangle,minimum height=1.8em,minimum width=2em]
\node [block](P){$\begin{matrix}P_{11}\rule[-1.3ex]{-1.3ex}{0pt}& P_{12}\\
P_{21}\rule{-1.3ex}{0pt} & G\end{matrix}$};
\node [block,below of=P](K){$K$};
\draw [<-] (P.east)+(0,-0.3) -- +(0.6,-0.3) |- node[pos=0.25]{$u$} (K);
\draw [<-] (P.east)+(0,0.3) -- +(1,0.3) node [anchor=west]{$w$};
\draw [->] (P.west)+(0,-0.3) -- +(-0.6,-0.3) |- node[swap,pos=0.25]{$y$} (K);
\draw [->] (P.west)+(0,0.3) -- +(-1,0.3) node [anchor=east]{$z$};
\end{tikzpicture}
\caption{Four-block plant with controller in feedback\label{fig:4blockplant}}
\end{figure}

\noindent For simplicity, assume for now that $S \subseteq M$. In Figure~\ref{fig:4blockplant}, the set of achievable closed-loop maps $w\mapsto z$ subject to $K$ belonging to some set $S$ is given by
\[
\mathcal{C} = \set{ P_{11} - P_{12}h(K)P_{21} }{K\in S}
\]
Selecting a controller that optimizes some closed-loop performance metric is equivalent to selecting the best $T\in \mathcal{C}$ and then finding the $K\in S$ that yields this $T$.

Roughly, the works~\cite{rotkowitz02,rotkowitz06} give a necessary and sufficient condition such that $h(S) = S$. If this condition holds, then $\mathcal{C} = \set{ P_{11} - P_{12} K P_{21} }{ K\in S }$, and so the set of achievable closed-loop maps is affine and easily searchable. The condition is called \eemph{quadratic invariance}, and a generic definition is given below.

\begin{defn}[Quadratic invariance]\label{def:qi}
We say that the set $S$ is quadratically invariant (QI) under $G$ if for all $K\in S$, we have $KGK\in S$.
\end{defn}

In~\cite{rotkowitz02}, the input and output spaces are Banach spaces, and $G$ and $K$ are bounded linear operators. In~\cite{rotkowitz06}, the extended spaces $\ell_{2e}$ and $L_{2e}$ are used and the associated maps are then continuous linear operators. In this work we use different spaces still, but the generic definition of quadratic invariance remains the same.

We now state the main results from~\cite{rotkowitz02,rotkowitz06}. Additional notation and terminology is defined after each theorem statement.

\begin{thm}[see~\cite{rotkowitz02}] \label{thm:qi_banach}
Suppose $\mathcal{Y}$ and $\mathcal{U}$ are Banach spaces, $G \in L(\mathcal{U},\mathcal{Y})$ and  $S \subseteq L(\mathcal{Y},\mathcal{U})$ is a closed subspace. Further suppose that $N\cap S = M\cap S$. Then
\[
S\text{ is QI with respect to }G
\quad \iff \quad
h(S\cap M) = S\cap M
\]
\end{thm}

In Theorem~\ref{thm:qi_banach}, $L(\cdot,\cdot)$ denotes the set of bounded linear operator from the first argument to the second. Also, $N$ is the set of $K\in L(\mathcal{Y},\mathcal{U})$ such that $1\in \rho_{\textup{uc}}(GK)$, the unbounded connected component of the resolvent set of $GK$. The condition that $N\cap S = M\cap S$ is admittedly technical in nature, but the result of the theorem is very simple; quadratic invariance is equivalent to $S$ being invariant under $h$.  

\begin{thm}[see~\cite{rotkowitz06}] \label{thm:qi_topological}
Suppose $G \in \mathcal{L}(L_{2e}^n,L_{2e}^m)$ and $S \subseteq \mathcal{L}(L_{2e}^m,L_{2e}^n)$ is an inert closed subspace. Then
\[
S\text{ is QI with respect to }G
\quad \iff \quad
h(S) = S
\]
\end{thm}

In Theorem~\ref{thm:qi_topological}, $\mathcal{L}(\cdot,\cdot)$ denotes the set of continuous linear maps from the first argument to the second. The requirement that $S$ be \emph{inert} means that the impulse response matrix of $GK$ must be entry-wise bounded over every finite time interval for all $K\in S$. Among other things, this technical condition guarantees that $I-GK$ is always invertible, and so $h$ is well-defined over all of $S$. An analogous result to Theorem~\ref{thm:qi_topological} in which $L_{2e}$ is replaced by $\ell_{2e}$ is also provided in~\cite{rotkowitz06}.

It is interesting to note that both sides of the equivalences proved in Theorems~\ref{thm:qi_banach} and~\ref{thm:qi_topological} are purely algebraic statements. In other words, they can be stated in terms of a finite number of algebraic operations (addition, multiplication, inversion). This seems at odds with the technical assumptions required in the theorems. For example, $S$ being a closed subspace means that $S$ should contain all of its limit points. This is an analytic concept requiring an underlying norm or a topology at the very least.

This observation is the starting point for this work, where we show
that \emph{invariance results} akin to Theorems~\ref{thm:qi_banach}
and~\ref{thm:qi_topological} can be obtained in a purely algebraic
setting without requiring anything more than well-defined addition,
multiplication, and inversion. This makes rings and fields the natural
objects to work with, and we will discuss them at greater length in
Section~\ref{sec:prelim}. In Section~\ref{sec:examples} we
give three specific settings where these  algebraic tools offer a
natural framework for modeling control systems. These are the cases
of sparse controllers, networks with delays, and multidimensional systems.

\section{The matrix case}
\label{sec:matrix_case}

The real matrix case is an example that illustrates when quadratic
invariance may be treated either analytically or algebraically. In
this section, we present the invariance result in the matrix case, and
give an outline of the proof using both the existing analytic
approach~\cite{rotkowitz02,rotkowitz06} and the algebraic approach
that is expanded upon in more detail in later sections of this
work. We present the proofs in sufficient detail to highlight the
mathematical machinery being used, but we skip over less relevant
details in the interest of clarity.

\begin{thm}[QI for matrices] \label{thm:qi_matrices}
Suppose $G\in \R^{m\times n}$ and $S\subseteq \R^{n\times m}$ is a subspace. Then the following holds.
\[
\text{$S$ is QI with respect to $G$}
\quad\iff\quad
h(S\cap M) = S\cap M
\]
\end{thm}

\begin{proof}
We outline a proof of the forward direction $(\implies)$. If $S$ is QI
with respect to $G$, then by definition we have $KGK\in S$ for all
$K\in S$. The first step is to show that $K(GK)^i \in S$ for
$i=1,2,\dots$ as well. This can be proven by induction using the
identity
\begin{equation*}
K(GK)^{i+1} = \frac12 \Bigl( \bigl(K+K(GK)^i\bigr)G\bigl(K+K(GK)^i\bigr) - KGK -\bigl( K(GK)^i\bigr)G\bigl(K(GK)^i\bigr) \Bigr)
\end{equation*}
Next, we examine the function $h(K) = K(I-GK)^{-1}$ when $K\in S\cap M$. It suffices to show that $h(K)\in S$, since $h$ is involutive. We prove this result first via an analytic approach similar to the one used in~\cite{rotkowitz02}, and then using an algebraic approach.

The \emph{analytic approach} is to use an infinite series expansion. For $\alpha \in \C$, we have the following convergent series.
\begin{align}\label{eq:infsum}
K(I- \alpha GK)^{-1} = \sum_{i=0}^\infty K(GK)^i \, \alpha^{i}
\quad\text{for $|\alpha| < \frac{1}{\|GK\|}$}
\end{align}
Since $K(GK)^i\in S$ for all $K\in S$, and $S$ is a finite-dimensional subspace and therefore closed, the infinite sum converges to an element of $S$. Using an analytic continuation argument~\cite{rotkowitz02}, one can show that $K(I-\alpha GK)^{-1}\in S$ for all $\alpha$ such that $\det(I - \alpha GK) \ne 0$. It then follows that $h(K) \in S$ for all $K\in S\cap M$, as required.

The \emph{algebraic approach} is to use a finite series expansion. Pick some $K\in S$ such that $(I-GK) \in \R^{m\times m}$ is invertible. By the Cayley-Hamilton theorem, there exist $p_0,\dots,p_{m-1} \in \R$ such that
\[
(I-GK)^{-1} = p_0 I + p_1 (I-GK) + \dots + p_{m-1} (I-GK)^{m-1}
\]
Expanding and collecting like powers of $GK$, we find that
\[
K(I-GK)^{-1} = q_0 K + q_1 KGK + \dots + q_{m-1} K(GK)^{m-1}
\]
for some $q_0,\dots,q_{m-1}\in \R$. Once again, $K(GK)^i \in S$ for all $K\in S$, so every term in this finite sum belongs to the subspace $S$ and therefore $K(I-GK)^{-1} \in S$. The difference is that we did not require an analytic continuation, nor did we make use of the fact that $S$ is closed.

The converse direction $(\impliedby)$ can also be proven either using an analytic argument as in~\cite{rotkowitz02}, or using an algebraic argument as we will develop in Section~\ref{sec:rings_and_fields}.
\end{proof}

Notice that
 we give two proofs of the forward direction of
 Theorem~\ref{thm:qi_matrices}, one analytic and the other
 algebraic. The analytic approach is based on convergence and hence
 depends on the topology. Since this particular result is only stated
 for matrices, the choice of topology does not matter, but for more
 general convolution operators on infinite signal spaces the topology
 has a significant effect on the applicability of the result and the
 technical machinery required to effect the proof.  However, the
 algebraic approach is much more simple, and only relies on addition,
 multiplication and inversion.

The main results of this paper, given in
Sections~\ref{sec:rings_and_fields} and~\ref{sec:rationals},
generalize and expand upon the algebraic approach used above to
matrices with entries that belong to a commutative ring $R$.

\section{Algebraic preliminaries}
\label{sec:prelim}

The fundamental algebraic properties we wish to capture are simply
addition and multiplication. This leads naturally to rings and fields,
which are a fundamental building block of abstract algebra. 
These concepts also provide the framework which is commonly used to
state many widely-used results in control theory. For example, the set
of real-rational transfer functions is a field, and the subset of
proper ones is a ring. This viewpoint can be extremely useful, for
example, when parameterizing all stabilizing
controllers~\cite{vidyasagar}.  We refer the reader
to~\cite{lang_undergraduate} for an introduction to these concepts.

We now explain some of the conventions used throughout this paper. The
integers, reals, rationals, and complex numbers are denoted by $\Z$,
$\R$, $\Q$, and $\C$, respectively.  We use $R$ to denote an arbitrary
commutative ring with identity, and $\F$ to denote an arbitrary
field. The additive and multiplicative identity elements of $R$ are
denoted $0_R$ and $1_R$, respectively, but we will often omit the
subscript when it is clear by context. An invertible element of $R$ is
called a \eemph{unit}, and the set of units is written as $U(R)$. We
write $R[x]$ and $\F(x)$ to respectively denote the ring of
polynomials and the field of rational functions in the indeterminate
$x$. If $H \subseteq R$ is an ideal, we write $H\ideal R$. If $M\ideal
R$ is maximal ideal, the associated quotient ring $R/M$ is a field,
and is called the \eemph{residue field}. Finally, we make use of the
notion of an $R$-module, which is the generalization of a vector space
when the scalars belong to a ring $R$ rather than a field $\F$.

In this paper, we consider finite matrices with elements that belong
to $R$. Much of the familiar linear algebra theory carries over to
this more general setting. We refer the reader to~\cite{curtis} for
introduction to abstract linear algebra. We write $R^{m\times n}$ to
mean the set of $m\times n$ matrices with entries in $R$. Matrix
multiplication between matrices of compatible dimensions is defined in
the standard way. When the matrices are square, we write $M_n(R) =
R^{n\times n}$, which is a ring. The identity matrix is denoted $I_R$,
that is, the matrix whose diagonal and off-diagonal entries are $1_R$
and $0_R$, respectively.

Many concepts from matrix theory carry over to the more general
setting. Specifically, if $A \in M_n(R)$, the determinant $\det :
M_n(R) \to R$ is defined by the conventional Laplace expansion. The
adjugate (or classical adjoint) $\adj : M_n(R) \to M_n(R)$ also makes
sense, as it is defined in terms of determinants of submatrices. The
fundamental identity for adjugates holds as well, namely
\begin{equation}\label{eq:fundamental_adjoint_property}
A \adj(A) = \adj(A) A = \det(A) I_R.
\end{equation}
The matrix $A \in M_n(R)$ is invertible if and only if $\det(A) \in U(R)$. In this case, the inverse is unique, and is given by
\[
 A^{-1} = \left( \det(A) \right)^{-1} \adj(A)
\]
For an introduction to the adjugate and associated results, we refer the reader to~\cite{robinson05}. We now state a fundamental result.

\begin{prop}\label{prop:ch}
Suppose $A \in M_n(R)$. Let $p_A \in R[x]$ be the characteristic polynomial of $A$, given by $p_A(x) = \det(A - x I)$. Suppose $p_A$ has the form
\[
p_A(x) = p_0 + p_1 x + \dots + p_n x^n
\]
Then the following equations hold in $M_n(R)$.
\begin{enumerate}[(i)]
\item $p_0 I + p_1 A + \dots + p_n A^n = 0$ \label{i:ch}
\item $\adj(A) = -( p_1 I + p_2 A + \dots + p_n A^{n-1} )$ \label{i:adj} 
\end{enumerate}
\end{prop}

\begin{proof}
See Section~\ref{sec:proofs}.
\end{proof}

Item~\eqref{i:ch} in Proposition~\ref{prop:ch} is commonly known as the Cayley-Hamilton theorem. This result plays an important role in our approach because it enables us to express quantities such as $(I-GK)^{-1}$ as finite sums.

\section{Invariance for rings and fields}\label{sec:rings_and_fields}

In this section, we take the notion of quadratic invariance discussed in the introduction and show how it fits into the framework of matrices over commutative rings or fields. These results generalize the algebraic invariance result for real matrices from Section~\ref{sec:matrix_case}. Complete proofs for all the results of this section are given in Section~\ref{sec:proofs}.

Our first main invariance result holds over matrices whose entries belong to an arbitrary commutative ring $R$ with identity. Terminology is explained after the theorem statement.

\begin{thm}[QI for rings]\label{thm:qi_rings}
Suppose $G\in R^{m\times n}$ and $S\subseteq R^{n\times m}$ is an $R$-module.
	\begin{enumerate}
	\item If $2_R\in U(R)$, then
	\begin{equation*}
		S\text{ is QI with respect to }G
		\quad \implies \quad
		K \adj(I_R - GK) \in S \text{ for all }K\in S
	\end{equation*}
	\item If every residue field of $R$ has at least $\min(m,n)+1$ elements, then
	\begin{equation*}
		S\text{ is QI with respect to }G
		\quad \impliedby \quad 
		K \adj(I_R - GK) \in S \text{ for all }K\in S
	\end{equation*}
	\end{enumerate}
\end{thm}

The notion of $R$-module is analogous to that of a subspace. That is, $S$ contains all linear combinations of its elements, where the linear combinations have coefficients in $R$.

Theorem~\ref{thm:qi_rings} contains several technical conditions which we will now explain. For the $(\implies)$ result, the condition $2_R \in U(R)$ means that $2_R := 1_R + 1_R$ must be a unit. We will now show that this condition is necessary by providing a counterexample. The condition is satisfied for example in the ring of rationals~$\Q$, but not in the ring of integers~$\Z$. Consider therefore the following integer example.
\[
 S =\set{\bmat{2x & y & z\\y & z & 0 \\ z & 0 & 0}}{x,y,z\in \Z},\quad
G = \bmat{0&0&0\\0&0&1\\0&1&0}
\]
It is straightforward to check that $S$ is a $\Z$-module, and is quadratically invariant with respect to $G$. Now consider the following particular element of $S$.
	\[
	K_0 = \bmat{0&0&1\\0&1&0\\1&0&0}\in S
	\]
	and so
	\[
	K_0\adj(I-GK_0) = \bmat{1&-1&1\\-1&1&0\\1&0&0} \notin S
	\]
Therefore the first part of Theorem~\ref{thm:qi_rings} does not hold. For more details on why this is so, refer to the proof of Theorem~\ref{thm:qi_rings} in Section~\ref{sec:proofs}.

For the $(\impliedby)$ result, a residue field is the field obtained by taking the quotient $R/M$ for some maximal ideal $M \ideal R$. For example, in the ring of integers $\Z$, the maximal ideals are the sets $p\Z$ for some prime $p$. So for $p=2$, the associated ideal is the set of even integers, and the residue field is $\Z/2\Z$; the integers modulo $2$. This field only has two elements and so the conditions of the theorem would not be satisfied for matrices with at least two rows or two columns.

We now specialize the above invariance result to fields. Axiomatically, a field is simply a ring for which every nonzero element is a unit. The results of this section hold for an arbitrary field~$\F$. We begin by stating the invariance result, and then we explain the differences between the field and ring cases. In particular, several concepts become simpler when the ring in question is a field. 

\begin{thm}[QI for fields]\label{thm:qi_fields}
Suppose $G\in \F^{m\times n}$ and $S\subseteq \F^{n\times m}$ is a subspace over~$\F$. Further suppose that $\F$ contains at least $2\min(m,n)+1$ distinct elements and $\ch\F \ne 2$.
\[
S\text{ is QI with respect to }G
\quad \iff \quad
h(S\cap M) = S\cap M
\]
where $M = \set{ K\in \F^{n\times m} }{ \det(I-GK)\ne 0 }$.
\end{thm}

The \emph{characteristic} of a field $\F$, denoted $\ch\F$ is the smallest $k$ such that $\overbrace{1 + \dots + 1}^{k\text{ times}} = 0$. When there is no such $k$, then we say $\ch\F=0$. Note that requiring $\ch\F\ne 2$ is the same as the condition that $2_R\in U(R)$ when the ring $R$ is specialized to the field $\F$.

In the real-number case $\F=\R$, both technical assumptions are always satisfied because $\R$ has infinitely many elements and $\ch\R=0$. We then precisely recover Theorem~\ref{thm:qi_matrices}. Note that Theorem~\ref{thm:qi_fields} may also be applied to finite fields such as~$\Z/p\Z$, the field of integers modulo~$p$. 

\section{Invariance for rationals}\label{sec:rationals}

In this section, we specialize the ring and field invariance results
of Section~\ref{sec:rings_and_fields} to rational functions in
multiple variables. This leads to quadratic invariance results without
any technical requirement on $S$ such as closure or the existence of
limits. As we shall see in Section~\ref{sec:examples}, this framework
can accommodate systems with delays or spatiotemporal systems.

Let $\R(x)$ be the set of rational functions in the indeterminate $x$
with coefficients in $\R$. Because $\R(x)$ is a field, we may apply
Theorem~\ref{thm:qi_fields}. We obtain the following result.

\begin{thm}[QI for rationals]\label{thm:qi_rationals}
Suppose $G\in\R(x)^{m\times n}$, and $S \subseteq \R(x)^{n\times m}$ is an $\R(x)$-module.
\[
S \text{ is QI with respect to } G \quad \iff\quad h(S\cap M) = S\cap M
\]
\end{thm}

Theorem~\ref{thm:qi_rationals} is the simplest algebraic result
for quadratic invariance of rational functions, and provides the
technical basis for the remainder of this paper. However, it is not
directly applicable to most control systems, because for physical
models one typically has the constraint that the system $G$ is
proper. This applies for example if $G$ is a
transfer function that represents a causal time-invariant system, and
we seek a controller that is also causal and time-invariant.

Let $\R(s)_{\p}$ be the set of proper rational functions in the
indeterminate $s$, and let $\R(s)_{\sp} \subset \R(s)_{\p}$ denote the
strictly proper rationals. Note that proper rationals are a ring
rather than a field, because the inverse of a proper rational function
is generally not proper. The result is given below.

\begin{thm}[QI for proper rationals]\label{thm:qi_rationals_proper}
Suppose $G\in\R(s)_{\sp}^{m\times n}$, and $S \subseteq \R(s)_{\p}^{n\times m}$ is an $\R(s)_{\p}$-module.
\[
S \text{ is QI with respect to } G \quad \iff\quad h(S) = S
\]
\end{thm}

We now extend the rational results of Theorems~\ref{thm:qi_rationals}
and~\ref{thm:qi_rationals_proper} to rational functions of multiple
variables with mixed properness constraints. This will allow
this approach to be applied to two additional classes of systems. The
first is networks of linear systems interconnected by delays, and the
second is multidimensional systems. These are discussed in Section~\ref{sec:examples}.
Specifically, define the sets of indeterminates $\mathbf{x} =
(x_1,\dots,x_\ell)$ and $\mathbf{s} = (s_1,\dots,s_k)$. We are
interested in the ring of multivariate rationals
$\R(\mathbf{x},\mathbf{s})$ where we have imposed a properness
constraint on each of the $s_i\in \mathbf{s}$. Note that this is not
the same as the rational function itself being proper. For example,
\[
g = \frac{s_1s_2s_3}{s_1^2+2s_2+s_3}
\]
is proper in each of the variables $s_1$, $s_2$, $s_3$, but is not
proper by the standard definition since the degree of the numerator is
larger than that of the denominator. We will use the subscripts p or
sp to apply individually to each of the $s_i\in\mathbf{s}$ while
ignoring the $x_i\in\mathbf{x}$. Therefore, $g \in \my_{\p}$. The
multivariate rational invariance result is given below.

\begin{thm}[QI for multivariate rationals]\label{thm:qi_mixed_rationals}
Suppose that $G \in \my^{m\times n}_{\sp}$, and $S \subseteq \my_{\p}^{n\times m}$ is an $\my_{\p}$-module.
\[
S \text{ is QI with respect to } G \quad\iff \quad h(S) = S
\]
\end{thm}

In Section~\ref{sec:examples}, we give an example of a class of
systems that can be represented by a multivariate rational function
with mixed properness constraints such as those used in
Theorem~\ref{thm:qi_mixed_rationals}.

The invariance results of Section~\ref{sec:rings_and_fields} only rely
on the algebraic properties of the objects involved, so our results
may be applied to a variety of examples beyond the ones mentioned in
this section. As a simple example, we may replace $\R$ by $\Q$ or $\C$
in Theorems~\ref{thm:qi_rationals}--\ref{thm:qi_mixed_rationals}.

\section{Examples}\label{sec:examples}

In this section, we show some examples of problems that can be modeled using
our algebraic framework. The purpose is to illustrate that the
constraint that $S$ be an $R_{\p}$-module occurs frequently and in a variety of
different situations.

\subsection{Sparse controllers}

The simplest class of systems that we can analyze are systems with rational
transfer functions subject to controllers with sparsity constraints. If every nonzero entry in the controller is required to be a proper
rational function in $\R(s)_{\p}$, it is clear that the set $S$ of admissible controllers is an $\R(s)_{\p}$-module.

\subsection{Network with delays}

Consider a distributed system where the subsystems affect one another
via delay constraints. We wish to design a decentralized controller subject to
communication delay constraints between subcontrollers.

Consider the simple example of two plants, each with their own controller. We represent the plants and their associated controllers by the transfer functions $(G_i(s),K_i(s))$. Suppose the controllers communicate with each other using a bilateral network that taxes all transmissions with a delay $d = e^{-s\tau}$. The example is illustrated in Figure~\ref{fig:networkexample} 
\begin{figure}[ht]
\centering
\begin{tikzpicture}[thick,auto,>=latex,node distance=10mm]
\def\dx{0.4}
\def\dh{0.15}
\tikzstyle{block}=[draw,rectangle,minimum height=1.8em,minimum width=2em]
\node [block](P1) at (-2,0) {$G_1$};
\node [block,below of=P1](K1){$K_1$};
\draw [<-] (P1.east) -- +(\dx,0) |- node[pos=0.25]{$u_1$} (K1);
\draw [->] (P1.west) -- +(-\dx,0) |- node[swap,pos=0.25]{$y_1$} (K1);
\node [block](P2) at (2,0) {$G_2$};
\node [block,below of=P2](K2){$K_2$};
\draw [<-] (P2.east) -- +(\dx,0) |- node[pos=0.25]{$u_2$} (K2);
\draw [->] (P2.west) -- +(-\dx,0) |- node[swap,pos=0.25]{$y_2$} (K2);
\node [block,minimum width=6cm](N) at (0,-2.2) {Network with delay $d$};
\coordinate[shift={(-\dh,0)}] (L1) at (K1.south);
\coordinate[shift={(\dh,0)}] (R1) at (K1.south);
\draw [->] (L1) -- (L1 |- N.north);
\draw [<-] (R1) -- (R1 |- N.north);
\coordinate[shift={(-\dh,0)}] (L2) at (K2.south);
\coordinate[shift={(\dh,0)}] (R2) at (K2.south);
\draw [->] (L2) -- (L2 |- N.north);
\draw [<-] (R2) -- (R2 |- N.north);
\end{tikzpicture}
\caption{Simple example consisting of two plants with networked controllers.\label{fig:networkexample}}
\end{figure}
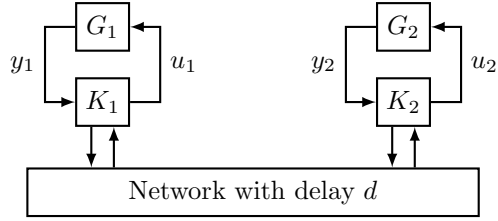
For simplicity, suppose each controller transmits all the measurements it receives to the other controller. Then the global plant and controller are characterized by the maps
\[
\bmat{y_1 \\ y_2} = \bmat{ G_1 & 0 \\ 0 & G_2}\! \bmat{u_1 \\ u_2}
\;\text{and}\;
\bmat{u_1 \\ u_2} = \bmat{K_{11} & K_{12}d \\ K_{21}d & K_{22}}\!\bmat{y_1 \\ y_2}
\]
Such an architecture is QI, and more detailed topologies of this type were studied in~\cite{rotkowitz06,rotkowitz10}. In general, the algebraic framework allows us to treat scenarios where the plant $G$ and controller $K$ are rational functions in $s$ and $d$. The constraint $K \in \R(s,d)_{\p}$, where
properness is enforced on $s$ and $d$ independently, naturally
guarantees that negative delays are forbidden, thus enforcing
causality.

Define the \emph{delay} of a transfer function as the difference between the
degree of $d$ in its denominator and numerator. For example,
\[
 \delay\left( \frac{1}{sd+2} \right) = 1
 \quad\text{and}\quad
 \delay\left( \frac{s+d^2}{s^2d+d^5}\right) = 3
\]
As a convention, $\delay(0) = \infty$. We can impose 
delay constraints on the controller using a set of the form
\[
 S = \set{K \in \R(s,d)_{\p}}{\delay(K_{ij}) \geq a_{ij}}
\]
where $a_{ij} \geq 0$ is the minimum delay (in multiples of $d$) between subcontrollers $i$
and $j$.  One can verify that $S$ is an $\R(s,d)_{\p}$-module, and so
we may apply Theorem~\ref{thm:qi_mixed_rationals} to deduce an
invariance condition. Similar results proved using very different
methods can be found in \cite{rotkowitz10}.

\subsection{Multidimensional systems}
\label{subsec:multid}

In many cases, we are interested in modeling and control of
systems whose states, inputs or outputs may be functions of a spatial
independent variable, such as in control of continuum mechanical
models such as fluids or elastic solids. While the dynamics may be
readily described by linear operator models, analysis of such models
has typically been performed using the tools of $C_0$ semigroup
theory~\cite{bamieh}.  Taking Fourier transforms spatially and Laplace
transforms temporally, one arrives at an algebraic formulation of such
systems~\cite{beck_2001,roesser}, where a linear system is described
by a rational function of two or more independent frequency variables.
The properness requirement that arises from causality is then only
required with respect to the temporal frequency variable.
This notion of multivariate transfer functions is used to represent
spatiotemporal dynamics in a variety of important papers
\cite{roesser, fornasini, beckdoyleglover, dulleruddandrea}.
Using this framework, we are able to circumvent the analytical
requirements and still explicitly construct the set of closed-loop
maps for such systems.  This class of systems is not addressed by
existing results on quadratic invariance, and in particular it is not
covered by the results of~\cite{rotkowitz06}. Our approach 
allows one to simply describe the set of achievable closed-loop maps
for such systems in both the centralized and decentralized cases.
In general, our framework allows us to consider transfer functions in
two sets of variables $R = \R(x_1,\dots,x_\ell,s_1,\dots,s_k)$ where
we impose properness on the $s_i$ but not on the $x_i$.

It is however worth noting that synthesis of the optimal controller
 for such multidimensional systems is a challenging problem, even in
the centralized case.  Therefore one cannot simply combine the results
in this paper with existing exact synthesis formulae or tools 
from the centralized case, and synthesize decentralized multidimensional controllers.

\section{Proofs of main results}\label{sec:proofs}

\begin{pproof}{Proposition~\ref{prop:ch}}
Apply the identity~\eqref{eq:fundamental_adjoint_property} to $A-xI$, which we view as an element of $M_n(R[x])$, and obtain
\begin{equation}\label{eq:chpf1}
 (A-xI)\adj(A-xI) = (p_0 + p_1 x + \dots + p_n x^n) I
\end{equation}
The adjugate is defined in terms of minors, so each entry of $\adj(A-xI)$ is an element of $R[x]$ of degree at most $n-1$. Because of the ring isomorphism $M_n(R[x]) \isomorphic (M_n(R))[x]$ (see~\cite[\S III.C]{curtis} for a proof), we may write $\adj(A-xI) = B_0 + B_1 x + \dots + B_{n-1} x^{n-1}$ for some $B_i\in M_n(R)$. Substituting into~\eqref{eq:chpf1} and viewing the result as as a polynomial identity in $(M_n(R))[x]$, we obtain
\begin{equation}\label{eq:chpf2}
(A - Ix) (B_0 + B_1 x + \dots + B_{n-1} x^{n-1}) 
	= p_0 I + p_1 I x + \dots + p_n I x^n
\end{equation}
Expanding~\eqref{eq:chpf2} and comparing coefficients, we obtain
\begin{align*}
AB_0 &= p_0 I \\
AB_k - B_{k-1} &= p_k I\qquad \text{for }k=1,\dots,n-1 \\
-B_{n-1} &= p_n I
\end{align*}
Left-multiplying each $p_k$ equation by $A^k$ and summing all equations, we obtain~\eqref{i:ch}. Similarly, left-multiplying each $p_k$ equation by $A^{k-1}$ for $k\ge 1$ and summing, we obtain
\[
p_1 + p_2 A + \dots + p_n A^{n-1} = -B_0 = -\adj(A)
\]
which is the statement of~\eqref{i:adj}.
\end{pproof}

\subsection{Invariance for rings}

In the following subsection, $R$ is an arbitrary commutative ring with identity. We begin with a lemma that allows us conclude that if $K\in S$ and $KGK\in S$, then we have $K(GK)^i\in S$ for $i=1,2,3,\dots$. 

\begin{lem}\label{lem:kgk}
Suppose $G\in R^{m\times n}$ and $S\subseteq R^{n\times m}$ is an
$R$-module. Further suppose that $2_R \in U(R)$. If $S$ is
quadratically invariant with respect to $G$, then for all $K\in S$:
\[
K(GK)^i \in S\qquad\text{for $i=1,2,\dots$}
\]
\end{lem}
\begin{proof}
The result follows by induction, using the identity:
\begin{equation*}
K(GK)^{i+1} = 2_R^{-1} \Bigl( \bigl(K+K(GK)^i\bigr)G\bigl(K+K(GK)^i\bigr) - KGK -\bigl( K(GK)^i\bigr)G\bigl(K(GK)^i\bigr) \Bigr)
\end{equation*}
where $2_R^{-1}$ is the multiplicative inverse of $2_R = 1_R+1_R$, which exists by assumption.
\end{proof}

Note that Lemma~\ref{lem:kgk} requires that $2_R\in U(R)$. As shown in the first example of Section~\ref{sec:rings_and_fields}, this requirement is necessary. If we strengthen the notion of quadratic invariance to instead require that $K_1 G K_2 \in S$ for all $K_1,K_2\in S$, then the conclusion of Lemma~\ref{lem:kgk} is trivial and the assumption $2_R\in U(R)$ is no longer required.

We will require an important property of polynomials. Specifically, we will need conditions under which a polynomial is uniquely specified by the values it takes on at finitely many points. For a real polynomial $f\in \R[x]$, the result is well-known. If we have $f(x)=0$ for all $x\in \R$, then $f=0$. In other words, all coefficients of $f$ are zeros and thus $f$ is the zero polynomial. The same result does not hold when we replace $\R$ by a general field $\F$ or by a commutative ring $R$. As a simple example, consider $f = x+x^2 \in (\Z/2\Z)[x]$, a polynomial with coefficients in the integers modulo~2. Then clearly $f(0)=f(1)=0$, yet $f\ne 0$.

Polynomials are closely related to \emph{Vandermonde matrices}, which we now define. The $N\times n$ Vandermonde matrix generated by $r_1,\dots,r_N \in R$ is defined as
\begin{equation}\label{eq:vandermonde}
V = \bmat{ 1 & r_1 & \dots & r_1^{n-1} \\
	   1 & r_2 & \dots & r_2^{n-1} \\
	   \vdots & \vdots & \ddots & \vdots \\
	   1 & r_N & \dots & r_N^{n-1} } \in R^{N\times n}
\end{equation}
So if $f(r) = a_0 + a_1 r + \dots + a_{n-1} r^{n-1}$, then the Vandermonde matrix~\eqref{eq:vandermonde} relates the values of $f$ evaluated at $r_1,\dots,r_N$ to the polynomial coefficients via
\[
\bmat{ f(r_1) \\ \vdots \\ f(r_N) }
= V \bmat{a_0 \\ \vdots \\ a_{n-1}}
\]
If there exists a left-invertible Vandermonde matrix $V$, then the coefficients $a_i$ are uniquely determined by the values $f(r_j)$. If $R$ is a field, we may assume without loss of generality that the Vandermonde matrix is square. The answer is given by the following proposition.

\begin{prop}\label{prop:inv_vandermonde_field}
Suppose $\F$ is a field. The following statements are equivalent
\begin{enumerate}[(i)]
\item The field $\F$ contains at least $n$ distinct elements.
\item\label{iii2} There exists an invertible $n\times n$ Vandermonde matrix.
\item\label{iii3} There exists a left-invertible $N\times n$ Vandermonde matrix for some $N\ge n$.
\end{enumerate}
\end{prop}

The proof follows from the following well-known formula for the determinant of a square Vandermonde matrix.
\[
\det V = \prod_{1 \le i < j \le n} (r_i-r_j)
\]
The result is  more complicated if~$R$ is a commutative ring with identity. In reference to Proposition~\ref{prop:inv_vandermonde_field}, there are cases where \eqref{iii3}$\nRightarrow$\eqref{iii2}. As an example, consider the ring $\Z[\beta]$ where $\beta = \tfrac12(1+\sqrt{-11})$.\label{exref} It is easy to check that the only units of $\Z[\beta]$ are $\pm 1$. There are no $3\times 3$ invertible Vandermonde matrices in this ring. To see why, note that if $V$ is generated by $x,y,z$ then $\det V = -(x-y)(y-z)(z-x)$, which is a unit if and only if each factor is a unit. Adding the factors together yields $\pm 1\pm 1\pm 1 = 0$, a contradiction. However, there exists a $4\times 3$ left-invertible Vandermonde matrix in this ring. For example,
\[
\bmat{ 1 & 1 & 1 & 1 \\
       0 & 1 & 2 & \beta \\
       0 & 1 & 4 & \beta^2 }
\bmat{ 1 & 1-10\beta & 1+\beta \\
       0 & 17+15\beta & -4-\beta \\
       0 & -11-3\beta & 2 \\
       0 & -7-2\beta & 1 }
= \bmat{1 & 0 & 0 \\ 0 & 1 & 0 \\ 0 & 0 & 1}
\]
The following lemma gives a complete characterization of left-invertibility for Vandermonde matrices in a general commutative ring with identity.
\begin{lem}\label{lem:inv_vandermonde}
Suppose $R$ is a commutative ring with identity. The following statements are equivalent.
\begin{enumerate}[(i)]
\item\label{itx:maxres}
Every residue field of $R$ has at least $n$ elements.
\item\label{itx:unideal}
The ideal generated by the determinants of all $n\times n$ Vandermonde matrices is equal to $R$, the unit ideal.
\item\label{itx:leftinv}
For some $N \ge n$, there exists a left-invertible $N\times n$ Vandermonde matrix.
\end{enumerate}
\end{lem}

\begin{proof}
We begin by showing \eqref{itx:maxres}$\iff$\eqref{itx:unideal}.
If \eqref{itx:maxres} is false, there exists some maximal ideal $M \ideal R$ such that the quotient ring $R/M$ contains fewer than $n$ elements. Therefore, given any set of elements $r_1,\dots, r_n \in R$, we must have $r_i-r_j \equiv 0 \pmod M$ for some $i,j$. It follows that if $V$ is the Vandermonde matrix generated by $r_1,\dots,r_n$ then it satisfies
\[
\det V = \prod_{1\le i < j\le n}(r_i-r_j) \equiv 0 \pmod M
\]
Therefore $\det V\in M$ for every Vandermonde matrix. So the ideal generated by all $n\times n$ Vandermonde determinants is contained in $M$, a proper ideal, and so cannot be the unit ideal. This shows that \eqref{itx:unideal} is false. Conversely, if \eqref{itx:unideal} is false then the ideal generated by all $n\times n$ Vandermonde determinants must be proper, and so is contained in some maximal ideal $M \ideal R$. In particular, $\det V \equiv 0 \pmod M$ for every $n\times n$ Vandermonde matrix $V$. Therefore, $\det V = 0$ for all $n\times n$ Vandermonde matrices $V$ with entries in $R/M$. Because $R/M$ is a field, every nonzero element is a unit. So $\det V$ is a unit of $R/M$ if and only if $V$ is generated by distinct elements. We conclude that $R/M$ must have fewer than $n$ distinct elements and so \eqref{itx:maxres} is false.

The result \eqref{itx:unideal}$\implies$\eqref{itx:leftinv} follows from the Cauchy-Binet formula, which gives the following expansion for $\det(LV)$ where $L$ and $V$ are not necessarily square but have a square product.
\[
\det(LV) = \sum_{\substack{s\subseteq\{1,\dots,N\}\\|s|=n}} \det(L_{:,s})\det(V_{s,:})
\]
Here, the sum is taken over all subsets of $\{1,\dots,N\}$ with $n$ elements. The corresponding columns and rows are extracted from $L$ and $V$ respectively and the associated determinants are multiplied together. Since each $V_{s,:}$ is an $n\times n$ Vandermonde determinant, if $L$ is a left-inverse for $V$, then $\det(LV) = \det(I) = 1$, and~\eqref{itx:unideal} follows. For the converse result~\eqref{itx:unideal}$\impliedby$\eqref{itx:leftinv}, a left-inverse can be explicitly constructed; see~\cite{robinson05} for a proof.
\end{proof}

\begin{lem}\label{lem:Helem}
Suppose $R$ is a commutative ring with identity, and $H$ is an $R$-module generated by $\{H_1,\dots,H_n\}$. Consider the following statements.
\begin{enumerate}[(i)]
\item\label{it:maxres}
Every residue field of $R$ has at least $n$ elements.
\item\label{it:Hgen}
Suppose $H$ is the $R$-module generated by some set $\{H_1,\dots,H_{n}\}$. Then $H$ is also generated by the set
$
\set{ H_1+rH_2 + \dots + r^{n-1} H_{n} }{ r\in R }
$
\end{enumerate}
Then \eqref{it:maxres}$\implies$\eqref{it:Hgen}. If the $H_i$ are also a basis for $H$, then \eqref{it:maxres}$\iff$\eqref{it:Hgen}.
\end{lem}
\begin{proof}
To prove~\eqref{it:Hgen}, it suffices to show that each $H_i$ is a linear combination of terms of the form $H_1+rH_1 + \dots + r^{n-1} H_{n}$. If~\eqref{it:maxres} holds, then by Lemma~\ref{lem:inv_vandermonde} there is some $N\times n$ left-invertible Vandermonde matrix $V\in R^{N\times n}$. Suppose $V$ is generated by $r_1,\dots,r_N$, and suppose $L\in R^{n\times N}$ is a left-inverse of $V$. Then it is straightforward to check that
\begin{equation}\label{eq:Hieqn}
H_i = \sum_{j=1}^N L_{ij} \bigl( H_1 + r_j H_2 + \dots + r_j^{n-1} H_{n} \bigr)
\qquad
\text{for }i=1,\dots,n
\end{equation}
as required. For the converse, suppose~\eqref{it:Hgen} holds. Then an equation of the form~\eqref{eq:Hieqn} must hold for some $L\in R^{n\times N}$ and $r_1,\dots, r_N\in R$. If the $H_i$ form a basis for $H$, then the coefficients corresponding to each $H_i$ in~\eqref{eq:Hieqn} must vanish. Therefore,
\[
\sum_{j=1}^N L_{ij} r_j^{k-1} = \begin{cases}
1 & i=k\\
0 & i\ne k
\end{cases}
\qquad\text{for }i,k=1,\dots,n
\]
In other words, $LV=I$, where $V$ is the (left-invertible) Vandermonde matrix generated by $r_1,\dots,r_N$. 
\end{proof}

Lemma~\ref{lem:Helem} may be specialized to polynomials and thus yields a sufficient condition under which a $f(r)=0$ for all $r\in R$ implies that $f$ is the zero polynomial.

\begin{cor}\label{cor:poly_zero}
Suppose $R$ is a commutative ring with identity and every residue field of $R$ has at least $n$ elements. Suppose $f \in R[x]$ and $\deg f \le n-1$. If $f(r)=0$ for all $r\in R$, then $f = 0$.
\end{cor}
\begin{proof}
Suppose $f(x) = a_0 + a_1 x + \dots + a_{n-1}x^{n-1}$. Applying Lemma~\ref{lem:Helem} to the module $H\subseteq R$ generated by $\{a_0,\dots,a_{n-1}\}$, we conclude that $H$ is generated by $\set{f(r)}{r\in R}$. But $f(r)=0$ for all $r\in R$, therefore $H = \{0\}$, and so $a_0=\dots=a_{n-1}=0$.
\end{proof}

We now have the tools we need to prove our main invariance result for rings.

\begin{pproof}{Theorem~\ref{thm:qi_rings}}
The result is trivial if $m=1$ or $n=1$. In this case, either $GK$ or $KG$ is scalar, so every $S$ is QI with respect to every $G$. Further, $K\adj(I_R-GK) = \adj(I_R-KG)K = K$, so the right-hand side always holds as well. We assume from now on that $m,n\ge 2$.

Suppose $S$ is QI with respect to $G$, and let $K\in S$. Using Proposition~\ref{prop:ch}, write:
\begin{align*}
K\adj(I_R-GK) &= -\sum_{i=1}^m p_i K(I_R-GK)^{i-1} \\
&= \sum_{i=1}^m h_i K(GK)^{i-1}
\end{align*}
where the $h_i\in R$ are obtained by expanding each $(I_R-GK)^{i-1}$
term and collecting like powers of $GK$. If $2_R\in U(R)$, then by Lemma \ref{lem:kgk} all terms in the sum are in $S$. Since $S$ is an $R$-module, it follows that $K\adj(I_R-GK) \in S$.

Conversely, suppose that $S$ is not QI with respect to $G$. Therefore, there exists some $K_0\in S$ such that $K_0 G K_0 \notin S$. We proceed by way of contradiction. Suppose that $K\adj(I_R-GK)\in S$ for all $K\in S$. In particular, it must hold for $K = rK_0$ with $r\in R$. Therefore, we conclude that
\[
r K_0 \adj(I_R - r GK_0) \in S\qquad\text{for all }r\in R
\]
Because each entry of the adjugate matrix is the determinant of an $(m-1)\times (m-1)$ minor, it follows that $\adj(I_R-rGK_0)$ is a polynomial in $r$ of degree at most $m-1$ with coefficients in $R^{m\times m}$. Letting $\adj(I_R-rGK_0)=\sum_{i=0}^{m-1} r^i B_i$, we obtain
\begin{equation*}
(K_0 B_0)r + (K_0 B_1)r^2 + \dots + (K_0 B_{m-1})r^m \in S\qquad
\text{for all }r\in R
\end{equation*}
If every residue field of $R$ has at least $m+1$ elements, we conclude from Lemma~\ref{lem:Helem} that $K_0 B_{i}\in S$ for $i=0,\dots, m-1$. Recall the identity~\eqref{eq:fundamental_adjoint_property}, and let $A = I_R-xGK_0 \in M_n(R[x])$. We then obtain
\begin{align*}
(I_R - xGK_0)(B_0 + xB_1 + \dots + x^{m-1}B_{m-1}) 
&= \det(I_R-xGK_0)I_R \\
&= (q_0 + q_1 x + \dots + q_m x^m) I_R
\end{align*}
where the $q_i\in R$ only depend on $G$ and $K_0$. Because of the ring isomorphism $M_m(R[x]) \isomorphic (M_m(R))[x]$, the equation above is a polynomial identity in $(M_m(R))[x]$. So we may collect like powers of $x$ and set all coefficients to zero. For the first two coefficients, we have $B_0 = q_0 I$ and $B_1 - GK_0 B_0 = q_1 I_R$. Note also that $B_0=I_R$, which follows because $B_0 = \adj(I_R) = I_R$. Therefore, $B_1 = GK_0 + q_1 I_R$. Based on our earlier conclusion that $K_0B_i \in S$, we have
$
K_0 B_1 = K_0 G K_0 + q_1 K_0 \in S
$.
Since $K_0 \in S$ by assumption and $S$ is an $R$-module, it follows that $K_0 G K_0 \in S$, a contradiction. If we use the identity $K\adj(I-GK) = \adj(I-KG)K$, we now have $\adj(I-KG)\in\R^{n\times n}$. Carrying out a similar argument, we deduce that the result also holds when every residue field of $R$ has at least $n+1$ elements, thus completing the proof.
\end{pproof}

The counterexample given in Section~\ref{sec:rings_and_fields} is an example for which $2_R\notin U(R)$. In that case, Theorem~\ref{thm:qi_rings} fails because Lemma~\ref{lem:kgk} fails. One way to avoid the technical requirement that $2_R\in U(R)$ is to strengthen the notion of quadratic invariance. For example, if we require that $K_1 G K_2 \in S$ for all $K_1, K_2\in S$ then Lemma~\ref{lem:kgk} follows without the requirement that $2_R\in U(R)$. Thus, we would obtain a weaker version of the first part of Theorem~\ref{thm:qi_rings}. The proof of Theorem~\ref{thm:qi_fields} uses similar machinery to that used in the proof of Theorem~\ref{thm:qi_rings}.

\begin{pproof}{Theorem~\ref{thm:qi_fields}}
As in Theorem~\ref{thm:qi_rings}, the cases $m=1$ and $n=1$ are trivial, so we assume $m,n\ge 2$. If $\ch\F \ne 2$, then $2_\F \ne 0$, which means $2_\F$ is invertible. Furthermore, $\F$ contains at least $\min(m,n)+1$ distinct elements, so we may apply both parts of Theorem~\ref{thm:qi_rings}. Therefore, $S$ being quadratically invariant with respect to $G$ is equivalent to $K\adj(I-GK)\in S$ for all $K\in S$. If we restrict to $K\in S\cap M$, then $\det(I-GK)\ne 0$. So we deduce that
\[
h(K) = -\bigl(\det(I-GK)\bigr)^{-1}K\adj(I-GK) \in S
\]
Since $h:M\to M$, it follows that $h(S\cap M) \subseteq S\cap M$, and by the involutive property of~$h$, we deduce that $h(S\cap M) = S\cap M$.

Conversely, suppose that $h(S\cap M) = S\cap M$. Then it follows that $K\adj(I-GK)\in S$ for all $K\in S\cap M$. Suppose $S$ is not quadratically invariant with respect to $G$, so there is some $K_0\in S$ such that $K_0 G K_0 \notin S$. Proceeding as in the proof of Theorem~\ref{thm:qi_rings}, let $\adj(I_R-rGK_0)=\sum_{i=0}^{m-1} r^i B_i$ and obtain
\begin{equation*}
(K_0 B_0)r + (K_0 B_1)r^2 + \dots + (K_0 B_{m-1})r^{m} \in S
\qquad
\text{for all }r\in R\text{ such that } rK_0 \in M
\end{equation*}
In other words, the inclusion holds whenever $\det(I-rGK_0)\ne 0$. This is a polynomial of degree $m$ so it has at most $m$ roots. Therefore, the constraint that $\det(I-rGK_0)\ne 0$ excludes at most $m$ elements of $\F$. Therefore, $\F$ must contain at least $2m+1$ elements so that there are at least $m+1$ elements remaining after all roots of $\det(I-rGK_0)$ have been excluded. As in the proof of Theorem~\ref{thm:qi_rings}, we may apply a similar argument to conclude that $\F$ contains at least $2n+1$ elements, and the proof in complete.
\end{pproof}

\subsection{Invariance for rationals}

Rational functions are fundamentally a field, so our first invariance result follows immediately from our invariance result for fields.

\begin{pproof}{Theorem~\ref{thm:qi_rationals}}
Note that $\mr$ is a field, so $S$ is a subspace over $\mr$. Furthermore, $\mr$ has infinitely many elements and $2 \ne 0$, so the result follows from Theorem~\ref{thm:qi_fields}.
\end{pproof}

The rationals become a ring when we impose a properness constraint. Furthermore, the strictly proper rationals are an ideal $\R(s)_{\sp} \ideal \R(s)_{\p}$. We can also check that this ideal is maximal, because the proper rationals that are not strictly proper are precisely the units of $\R(s)_{\p}$. Indeed, $\R(s)_{\sp}$ is the unique maximal ideal of $\R_{\p}$, so $\R_{\p}$ is a \emph{local ring}.

We may use these structural facts to prove the following lemma, which gives a condition that guarantees the invertibility of $I-GK$.

\begin{lem}\label{lem:invertible}
Suppose $G \in \R(s)_{\sp}^{m\times n}$ and $K \in \R(s)^{n\times m}_{\p}$. Then
$(I-GK)$ is invertible, and $(I-GK)^{-1} \in \R(s)_{\p}^{m\times m}$.
\end{lem}
\begin{proof}
Since $\R(s)_{\sp} \ideal \R(s)_{\p}$ is maximal and $G \in \R(s)_{\sp}^{m\times n}$, we have
\[
\det(I-GK) \equiv \det(I) \equiv 1 \pmod {\R(s)_{\sp}}
\]
It follows that $\det(I-GK) \notin \R(s)_{\sp}$, and is therefore a unit. So $I-GK$ is invertible.
\end{proof}

The motivation for choosing a strictly proper $G$ and a proper $K$ is inspired by classical feedback control. If we think of the proper rationals as transfer functions, a strictly proper $K$ means that the controller has no direct feedthrough term, a common assumption that ensures well-posedness of the closed-loop interconnection. We now present the proof of the invariance result for proper rationals.

\begin{pproof}{Theorem~\ref{thm:qi_rationals_proper}}
Note that $\R(s)_{\p}$ is a ring and $2 \in \R(s)_{\p}$. Because $\R(s)_{\p}$ has a unique maximal ideal $\R(s)_{\sp}$, the only residue field is $\R(s)_{\p} / \R(s)_{\sp} = \R$, which has infinitely many elements. Applying Theorem~\ref{thm:qi_rings}, we conclude that $S$ is QI with respect to $G$ if and only if $K\adj(I-GK)\in S$ for all $K\in S$. However, $I-GK$ is always invertible by Lemma~\ref{lem:invertible}, so $K\adj(I-GK)\in S$ if and only if $h(K)\in S$. By the involutive property of $h$, this is equivalent to $h(S)=S$.
\end{pproof}

\begin{pproof}{Theorem~\ref{thm:qi_mixed_rationals}}
First, note that $\my_{\p} \isomorphic \mx_{\p}$, so we may think of the multivariate transfer functions as proper transfer functions in $s_1,\dots,s_k$ with coefficients that are rational functions of $x_1,\dots,x_\ell$. As in Theorem~\ref{thm:qi_rationals_proper}, we still have $2\in \mx_{\p}$, but there is no longer a unique maximal ideal. Indeed, the maximal ideals are the sets:
\begin{align*}
M_i &= \set{f \in \mx_{\p} }{ f \text{ is strictly proper in }s_i } \\
&= \Bigl(\bigl( \R(\mathbf{x}) \bigr) (\mathbf{s}\setminus s_i)_{\p}\Bigr)(s_i)_{\sp}
\end{align*}
and it is easy to check that the corresponding residue fields $\mx_\p / M_i$ each have infinitely many elements. Furthermore, $\det(I-GK) \equiv \det(I) \equiv 1 \pmod{M_i}$ for each $i$, as in the proof of Lemma~\ref{lem:invertible}. So $\det(I-GK)$ does not belong to any maximum ideal and must therefore be a unit. The rest follows as in the proof of Theorem~\ref{thm:qi_rationals_proper}.
\end{pproof}

\section{Summary}
\label{sec:conclusion}

In this paper, we give an algebraic treatment of quadratic invariance, the well-known condition under which decentralized control synthesis may be reduced to a convex optimization problem. Our results hold for commutative rings with identity, and in particular specialize to the natural system-theoretic case of proper rational functions in one variable, as well as multidimensional rational functions. This formulation has the advantage of avoiding some of the technicalities in analytic treatments. In particular, notions of topology, limits, or norms are not required. Thus, quadratic invariance may be viewed as a purely algebraic concept.

\appendices
\section*{Acknowledgments}
The proof of Lemma~\ref{lem:inv_vandermonde} is due to Thomas Goodwillie, and the noninvertible Vandermonde example on Page~\pageref{exref} is due to David Speyer.

\bibliographystyle{abbrv}
\bibliography{alqi}


\end{document}